\documentclass{llncs}
\usepackage{float}
\floatstyle{ruled}
\usepackage{boxedminipage}
\usepackage{graphicx, subfigure}
\usepackage{epsfig}
\usepackage{theorem}
\usepackage{latexsym}
\usepackage{amssymb}
\newfloat{algorithm}{tbp}{loa}
\floatname{algorithm}{Algorithm}

\renewenvironment{proof}{{\it Proof. } }{{\hfill $\Box$}\vspace{.5pc}}

\newtheorem{specification}{Specification}
\newtheorem{lem}{Lemma}
\newcommand{\keywords}[1]{\par\addvspace\baselineskip
\noindent\keywordname\enspace\ignorespaces#1}

\newcommand{\eg}{\emph{e.g., }}
\newcommand{\ie}{\emph{i.e., }}

\begin{document}
\frontmatter          % for the preliminaries
\pagestyle{headings}  % switches on printing of running heads
\addtocmark{Deaf, Dumb, and Chatting Asynchronous Robots: 
        Enabling Distributed Computation and Fault-Tolerance Among Stigmergic Robot} % additional mark in the TOC

\mainmatter              % start of the contributions
\title{Snap-Stabilizing Linear Message Forwarding%
        \thanks{This work is supported by ANR SPADES grant.}
}

\author{
  Alain Cournier\inst{1} \and  
  Swan Dubois\inst{2} \and 
  Anissa Lamani\inst{1} \and 
  Franck Petit\inst{2} \and 
  Vincent Villain\inst{1} 
}

\authorrunning{Cournier {\em et al}.}   % abbreviated author list (for running head)
%
%% list of authors for the TOC (use if author list has to be modified)
\tocauthor{ Alain Cournier, Swan Dubois, Anissa Lamani, Franck Petit, Vincent Villain}
\institute{
MIS, Universit\'e of Picardie Jules Verne, France 
\and
LiP6/CNRS/INRIA-REGAL, Universit\'e Pierre et Marie Curie - Paris 6, France 
}

\maketitle
%%%%%%%%%%%%%%%%%%%%%%%%%%%%%%%%  ABSTRACT    %%%%%%%%%%%%%%%%%%%%%%%%%%%%%%%%%%%%%%
\begin {abstract}

In this paper, we present the first snap-stabilizing message forwarding protocol that uses 
a number of buffers per node being independent of any global parameter, that is $4$ buffers per 
link.  The protocol works on a linear chain of nodes, that is possibly an overlay on a large-scale and dynamic system,
\eg Peer-to-Peer systems, Grids\ldots  
Provided that the topology remains a linear chain and that nodes join and leave 
``neatly'', the protocol tolerates topology changes.  
We expect that this protocol will be the base to get similar results on more general topologies. 
%By ``{\it neatly}'', we mean that when a node leaves the
%system, it makes sure that the messages it has to send are transmitted, \ie all its buffers are free.    
%Thus, this is a preliminary version aiming to get similar results on more general topologies
%is the first snap-stabilizing, scalable message forwarding protocol that tolerates dynamicity.

\end {abstract}

\section{Introduction}

%\paragraph{Context.}
%%%%%
These last few years have seen the development of large-scale distributed systems.  
Peer-to-peer (P2P) architectures belong to this category.  They usually offer computational services or storage facilities.
Two of the most challenging issues in the development of such large-scale distributed 
systems are to come up with scalability and dynamicity.  {\em Scalability} is achieved by designing protocols 
with performances growing sub-linearly with the number of nodes (or, processors, participants).  {\em Dynamicity} refers to 
distributed systems in which topological changes can occur, \ie nodes may join or leave the system.

{\em Self-stabilization}~\cite{D00} is a general technique to design distributed systems that 
can tolerate arbitrary transient faults.  Self-stabilization is also well-known to be suitable for dynamic systems. 
This is particularly relevant whenever the distributed (self-stabilizing) protocol does not require any global parameters,
like the number of nodes ($n$) or the diameter ($D$) of the network.  
With such a self-stabilizing protocol, it is not required to change 
global parameters in the program ($n$, $D$, etc) when nodes join or leave the system.  Note that this property 
is also very desirable to achieve scalability. 

The {\em end-to-end communication} problem consists in delivery in finite time across the network of a sequence of 
data items generated at a node called the sender, to a designated node called the receiver.
This problem is generally split into the two following problems:
($i$) the {\em routing} problem, \ie the determination of the path followed by the messages to reach their destinations;
($ii$) the {\em message forwarding} problem that consists in the management of network resources in order to forward messages. 
The former problem is strongly related to the problem of spanning tree construction. %, especially the Breadth-First Search (BFS) one.  
Numerous self-stabilizing solutions exist for this problem, \eg \cite{HC92,KK05,JT03}. 

In this paper, we concentrate on the latter problem, \ie the message forwarding problem.   
More precisely, it consists in the design of a protocol managing the mechanism allowing
the message to move from a node to another on the path from the sender $A$ to the receiver $B$.
To enable such a mechanism, each node on the path from $A$ to $B$ has a reserved memory space called buffers.
With a finite number of buffers, the message forwarding problem consists in avoiding deadlocks and livelocks (even
assuming correct routing table).
Self-stabilizing solutions for the message forwarding problem are proposed in \cite{APV96,KOR95}. 
Our goal is to provide a snap-stabilizing solution for this problem.
A {\em snap-stabilizing protocol}~\cite{Bui07}
guarantees that, starting from any configuration, it always behaves
according to its specification, \ie it is a
self-stabilizing algorithm which is optimal in terms of stabilization
time since it stabilizes in $0$ steps.
Considering the message-forwarding problem, combined with a self-stabilizing routing protocol, 
snap-stabilization brings the desirable property that 
every message sent by the sender is delivered in finite time to the receiver.
By contrast, any self-stabilizing (but not snap-stabilizing) solution for this problem
ensures the same property, ``eventually''. %eventually, every message sent by the sender is delivered in finite time to the receiver.

The problem of minimizing the number of required buffers on each node is a crucial issue
for both dynamicity and scalability. 
The first snap-stabilizing solution for this problem can be found in~\cite{CDV09-1}.
Using $n$ buffers per node, this solution is not suitable for large-scale system. 
The number of buffers is reduced to $D$ in~\cite{CDV09-2}, which improves the scalability aspect.  However,
it works by reserving the entire sequence of buffers leading from the sender to the receiver.
Furthermore, to tolerate dynamicity, each time a topology change occurs in the system, 
both of them would have to rebuild required data structures, maybe on the cost of loosing 
the snap-stabilisation property.

In this paper, we present a snap-stabilizing message forwarding protocol that uses 
a number of buffers per node being independent of any global parameter, that is $4$ buffers per 
link. The protocol works on a linear 
chain of nodes, that is possibly an overlay on a large-scale and dynamic system \eg Peer-to-Peer systems, Grids\ldots  
Provided that ($i$) the topology remains a linear chain and ($ii$) that nodes join and leave 
``neatly'', the protocol tolerates topology changes.  By ``{\it neatly}'', we mean that when a node leaves the
system, it makes sure that the messages it has to send are transmitted, \ie all its buffers are free.    
We expect that this protocol will be the base to get similar results on more general topologies. 
%Thus, this preliminary version aiming to get similar results on more general topologies (or large-scale overlays)
%is the first snap-stabilizing, scalable message forwarding protocol that tolerates dynamicity. 

The paper is structured as follow: In Section~\ref{sec:model}, we define our model and some useful terms that are used afterwards. 
In Section~\ref{sec:algo}, we first give an informal overview of our algorithm, followed by its formal description. 
In Section~\ref{sec:proof}, we prove the correctness of our algorithm.
Dynamicity is discussed in Section~\ref{sec:dynamicity}. We conclude the paper in Section~\ref{sec:conclu}.  

%%%%%%%%%%%%%%%%%%%%%%%%%%%% Model and definition %%%%%%%%%%%%%%%%%%%%%%%%%%%%%%%%%%%%%%%
\section{Model and definitions\label{sec:model}}

\paragraph{\textbf{Network}.} 
We consider a network as an undirected connected graph $G=(V,E)$ where $V$ is the set of nodes (processors) and $E$ is the set of bidirectional communication links. A link $(p,q)$ exists if and only if the two processors $p$ and $q$ are neighbours. Note that, every processor is able to distinguish all its links. To simplify the presentation we refer to the link $(p,q)$ by the label $q$ in the code of $p$. In our case we consider that the network is a chain of $n$ processors.\\

\paragraph{\textbf{Computational model}.}
We consider in our work the classical local shared memory model introduced by Dijkstra~\cite{D74} known as the state model. In this model communications between neighbours are modelled by direct reading of variables instead of exchange of messages. The program of every processor consists in a set of shared variables (henceforth referred to as variable) and a finite number of actions. Each processor can write in its own variables and read its own variables and those of its neighbours. Each action is constituted as follow:

\begin{center} $<Label>::<Guard>$ $\rightarrow$ $<Statement>$ \end{center}

The guard of an action is a boolean expression involving the variables of $p$ and its neighbours. The statement is an action which updates one or more variables of $p$. Note that an action can be executed only if its guard is true. Each execution is decomposed into steps. 

The state of a processor is defined by the value of its variables. The state of a system is the product of the states of all processors. The local state refers to the state of a processor and the global state to the state of the system.  

Let $y$ $\in$ $C$ and $A$ an action of $p$ ($p$ $\in$ $V$). $A$ is {\em enabled} for $p$ in $y$ if and only if the guard of
$A$ is satisfied by $p$ in $y$. Processor $p$ is enabled in $y$ if and only if at least one action is enabled at $p$
in $y$. Let $P$ be a distributed protocol which is a collection of binary transition relations denoted by
$\rightarrow$, on $C$. An execution of a protocol $P$ is a maximal sequence of configurations  $e=
y_{0}y_{1}...y_{i}y_{i+1} \ldots$ such that, $\forall$ $i\ge0$, $y_{i} \rightarrow y_{i+1}$ (called a step) if $y_{i+1}$
exists, else $y_{i}$ is a terminal configuration. {\em Maximality} means that the sequence is either finite (and no action
of $P$ is enabled in the terminal configuration) or infinite. All executions considered here are assumed to be
maximal. $\xi$ is the set of all executions of $P$. 
Each step consists on two sequential phases atomically executed:
($i$) Every processor evaluates its guard;
($ii$) One or more enabled processors execute its enabled actions. 
When the two phases are done, the next step begins. 
This execution model is known as the \emph{distributed daemon}~\cite{BGM89}. 
We assume that the daemon is \emph{weakly fair}, meaning that if a processor $p$ is continuously $enabled$, 
then $p$ will be eventually chosen by the daemon to execute an action.

In this paper, we use a composition of protocols.  We assume that the above statement ($ii$) is applicable to every
protocol. In other words, each time an enabled processor $p$ is selected by the daemon, $p$ executes the enabled actions of every protocol.
%We consider that any processor $p$ is neutralized in the step $y_{i}$ $\rightarrow$
%$y_{i+1}$ if $p$ was enabled in $y_{i}$ and not enabled in $y_{i+1}$, but did not execute any action in $y_{i}
%\rightarrow y_{i+1}$.

%To compute the time complexity, we use the definition of round introduced in 10 and then modified in 3. This definition captures the execution rate of the slowest processor in any execution. The first round of $\Gamma \in \xi$, noted $\Gamma'$, is the minimal prefix of $\Gamma$ containing the execution of one action or the neutralization of every enabled processor from the initial configuration. Let$\Gamma''$ be the suffix of $\Gamma$ such that $\Gamma$ = $\Gamma'$ $\Gamma''$. The second round of $\Gamma$ is the first round of $\Gamma''$, and so on.

%\section{Stabilization}

%We give in this section a formal definition of self-stabilization and snap stabilization respectively:\\

%\subsection{Self-Stabilization} Let $\Gamma$ be a task, and $S_{\Gamma}$ a specification of $\Gamma$. A protocol $P$ is self-stabilizing for $S_{\Gamma}$	if and only if $\forall \Gamma \in \xi$, there exists a finite prefix $\Gamma'$ = ($\gamma_{0}$, $\gamma_{1}$, ..., $\gamma_{l}$) of $\Gamma$ such that any executions starting from $\gamma_{l}$ satisfies $S_{\Gamma}$.\\

\paragraph{\textbf{Snap-Stabilization}.} 
Let $\Gamma$ be a task, and $S_{\Gamma}$ a specification of $\Gamma$ . A protocol $P$ is snap-stabilizing for $S_{\Gamma}$ if and only if $\forall \Gamma \in \xi$, $\Gamma$ satisfies $S_{\Gamma}$.

%%%%%%%%%%%%%%%%%%%%%%%%%%%%%%%%%%%%%%%%%%%%%%%%%%%%%%%%%%%%%%%%%%%%%
%\section{Snap stabilizing message forwarding specification}
%\section{Snap stabilizing message forwarding specification}

\paragraph{\textbf{Message Forwarding Problem}.}

%Messages transit in the network in the Store and Forward model \ie they are stored temporally in each processor before being transmitted. Once the message is transmitted it can be deleted from the previous processor. Note that in order to store messages, each processor use a space memory called buffer. We assume in our case that each buffer can store a whole message and each message needs only one buffer to be stored. It is clear that each processor uses a finite number of buffers for the message forwarding. Thus the aim is to bound these resources avoiding deadlocks (a configuration is which in every execution some messages can not be transmitted) and starvation (a configuration from which, in every execution, some processors are no longer able to generate messages). Thus some control mechanisms must be introduced in order to avoid these kind of situations. Note that an algorithm that prevent deadlock situations is said to be a deadlock free controller. We say that an algorithm solve the forwarding problem if and only if it responds to the following specification: 

The message forwarding problem is specified as follows: %Specification~\ref{spec:SP}.

\begin{specification}[$SP$]\label{spec:SP}
%\textbf{\textit{Specification 1 (SP)}}
A protocol $P$ satisfies $SP$ if and only if the following two requirements are satisfied in every execution of $P$:
\begin{enumerate}
\item Any message can be generated in a finite time;
\item Any valid message is delivered to its destination once and only once in a finite time.
\end{enumerate}
\end{specification}

%In order to prove the correctness of our algorithm, we will need Specification~$SP'$, which is 
%a lighten version of Specification~\ref{spec:SP} such that the second condition is specified 
%as follows: 
%
%\begin{specification}[$SP'$]\label{spec:SPprime}
%A protocol $P$ satisfies $SP$ if and only if the following two requirements are satisfied in every execution of $P$:
%\begin{enumerate}
%\item Any message can be generated in a finite time;
%\item 
%Any valid message is deliver to its destination in a finite time. In other words, Specification~$SP'$ is similar to 
%Specification~$SP$, except that message duplications can occur. 
%\end{enumerate}
%\end{specification}

%%%%%%%%%%%%%%%%%%%%%%%%%%%%%%%%%%%%%%%%%%%%%%%%%%%%%%%%%%%%%%%%%%%%%
 \paragraph{\textbf{Buffer Graph}} 

%\fixme{en faire une sous-section de la section suivante ou precedente.}

In order to conceive our snap stabilizing algorithm we will use a structure called Buffer Graph introduced in \cite{MS78}. A Buffer Graph is defined as a directed graph where nodes are a subset of the buffers of the network and links are arcs connecting some pairs of buffers, indicating permitted message flow from one buffer to another one. Arcs are permitted only between buffers in the same node, or between buffers in distinct nodes which are connected by communication link.
%\fixme{en faire une sous-section de la section suivante ou precedente.}
 
Let us define our buffer graph (refer to Figure~\ref{BG}):

\begin{figure}
   \centering
   \includegraphics[width=9cm]{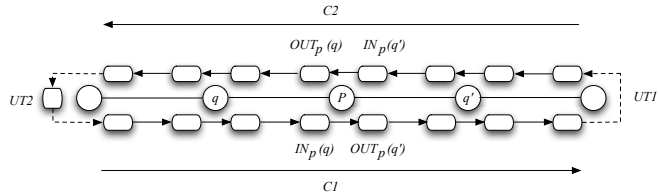}
   \caption{Buffer Graph}\label{BG}
\end{figure} 
 
%\fixme{deadlocks, $BG$ ou $C_BG$ non d\'efinis.}

\noindent
Each processor $p$ has four buffers, two for each link $(p,q)$ such as $q \in N_{p}$ (except for the processors that
are at the extremity of the chain that have only two buffers, since they have only one link). Each processor has two
input buffers denoted by $IN_{p}(q)$, $IN_{p}(q')$ and two output buffers denoted by $OUT_{p}(q)$, $OUT_{p}(q')$ such as $q,q'\in
N_{p}$ and $q\ne q'$ (one for each neighbour). The generation of a message is always done in the output buffer of the
link $(p,q)$ so that, according to the routing tables, $q$ is the next processor for the message in order to reach 
the destination.  Let us refer to $nb(m,b)$ as the next buffer of Message~$m$ stored in $b$,  
$b \in \{IN_{p}(q) \vee OUT_{p}(q)\}$, $q\in N_{p}$. We have the following properties:
\begin{enumerate}
\item $nb(m,IN_{p}(q))= OUT_{q}(p)$
\item $nb(m,OUT_{p}(q)= IN_{q}(p)$
\end{enumerate}  

%%%%%%%%%%%%%
%Any distributed protocol $P$ satisfying Specification~$SP$ prevents the system of \emph{starvation}, \ie a configuration
%from which, in every execution, some processors are no longer able to generate messages.   

%%%%%%%%%%%%%%%

%It has been proved in \cite{MS78} that $P$ is deadlock-free if and only if $BG$ is a directed acyclic graph (DAG). 
%One can easily observe that the buffer graph $BG$ is a directed graph without cycles since it is a directed chain. 
%Thus, we can define a controller of $CBG$ that is a free-deadlock controller.   

\section{Message Forwarding\label{sec:algo}}

In this section, we first give the idea of our snap stabilizing message forwarding algorithm in the informal overview, 
then we give the formal description followed by the correctness proofs. 

\subsection{Overview of the algorithm}

In this section, we provide an informal description of our snap stabilizing message forwarding algorithm that 
tolerates the corruption of the routing tables in the initial configuration. 

To ease the reading of the section, we assume that there is no message in the system whose the destination is not
in the system.  This restriction is not a problem as we will see in Section~\ref{sec:dynamicity}.

We assume that there is a self-stabilizing algorithm, $Rtables$, that calculates the routing tables and runs simultaneously to our
algorithm. %Moreover, if there is one action that is enabled in each algorithm on the same processor $p$ at the same
%time then $p$ executes both actions. In another hand, 
We assume that our algorithm has access to the routing tables
via the function $Next_{p}(d)$ which returns the identity of the neighbour to which $p$ must forward the message to
reach the destination $d$. To reach our purpose we define a buffer graph on the chain which consists of two chains,
one in each direction ($C1$ and $C2$ refer to Figure \ref{BG}). 

The overall idea of the algorithm is as follows: When a processor wants to generate a message, it consults the routing
tables to determine the next neighbour by which the message will transit in order to reach the destination. Note that
the generation is always done in the Output buffers. Once the message is on the chain, it follows the buffer chain
(according to the direction of the buffer graph) and if the messages can progress enough in the system (move) then it
will either meet its destination and hence it will be consumed in a finite time or it will reach the input buffer of
one of the processors that are at the extremity of the chain. In the latter case, if the processor that is at the
extremity of the chain is not the destination then, that means that the message was in the wrong direction. The idea
is to change the direction of the message by copying it in the output buffer of the same processor (directly (UT1) or
using the extra buffer (UT2), refer to Figure~\ref{BG}). 

Note that if the routing tables are stabilized and if all the messages are in the right direction then all the
messages can move on $C1$ or $C2$ only and no deadlock happens. However, in the opposite case (the routing tables are not stabilized or
some messages are in the wrong direction), deadlocks may happen if no control is introduced. For instance, suppose
that in the initial configuration all the buffers, uncluding the extra buffer of $UT2$, contain different messages such that 
no message can be consumed. It is clear that in this case no message can move and the system is deadlocked. Thus in order to solve this problem we
have to delete at least one message. However, since we want a snap stabilizing solution we cannot delete a message
that has been generated. Thus we have to introduce some control mechanisms in order to avoid this situation to appear
dynamically (after the first configuration). In our case we decided to use the PIF algorithm that comprises two main
phases: Broadcast (Flooding phase) and Feedback (acknowledgement phase) to control and avoid deadlock situations.

Before we explain how the PIF algorithm is used, let us focus on the message progression again.
A buffer is said to be {\em free} if and only if it is empty (it contains no message) or contains the same message as
the input buffer before it in the buffer graph buffer. For instance, if $IN_{p}(q)=OUT_{q}(p)$ then $OUT_{q}(p)$ is a
free buffer. In the opposite case, a buffer is said to {\em busy}.
The transmission of messages produces the filling and the cleaning of each buffer, \ie each buffer is alternatively
free and busy. This mechanism clearly induces that {\em free slots} move into the buffer graph, a free slot
corresponding to a free buffer at a given instant. The moving of free slots is shown in Figure \ref{EmptySlot}\footnote{Note that in the
algorithm, the actions $(b)$ and $(c)$ are executed in the same step (refer to the guarded action $R3$).}.
Notice that the free slots move in the opposite direction of the message progression.  This is the key feature on
which the PIF control is based.

\begin{figure}
 \begin{minipage}[b]{.46\linewidth}
  \begin{center}
  \epsfig{figure=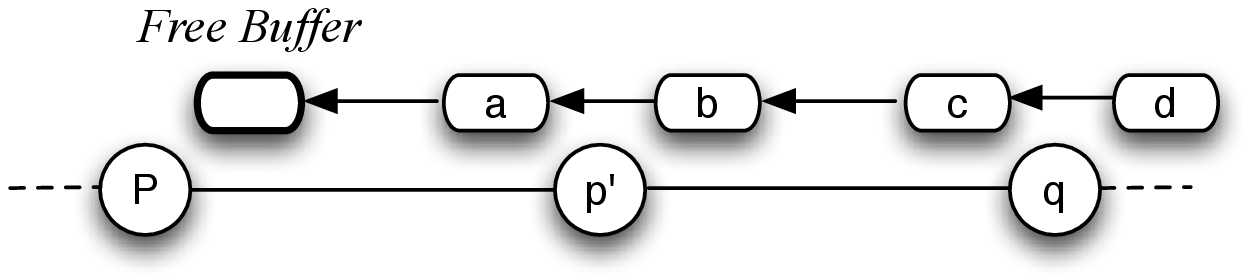,width=5.5cm}\\
  \textit{(a)} The input buffer of $p$ is free. Node $p$ can copy the message $a$. 
  \end{center}
 \end{minipage} \hfill
 \begin{minipage}[b]{.46\linewidth}
 \begin{center}
 \epsfig{figure=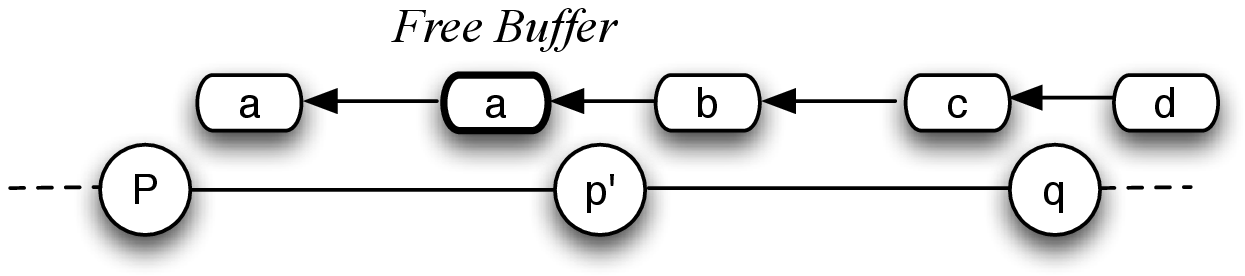,width=5.5cm}\\
  \textit{(b)} The output buffer of $p'$ is free. Node $p'$ can copy the message $b$.
  \end{center}
 \end{minipage}\hfill
 \begin{minipage}[b]{.46\linewidth}
\begin{center}  
  \epsfig{figure=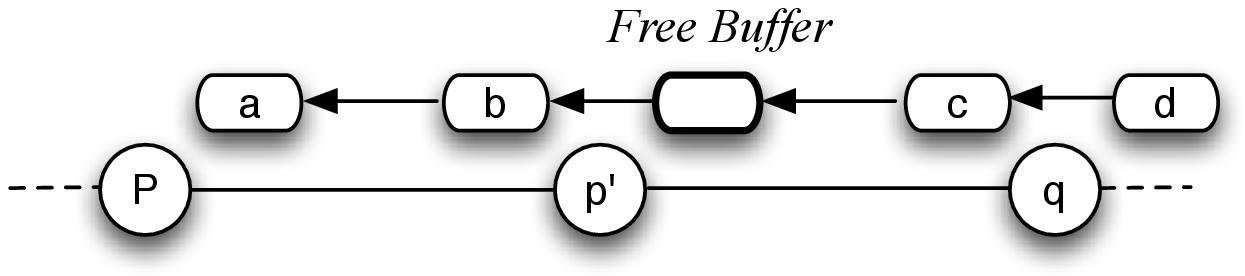,width=5.5cm}\\
  \textit{(c)} The input buffer of $p'$ is free. Node $p'$ can copy the message $c$.
  \end{center}
 \end{minipage}\hfill
 \begin{minipage}[b]{.46\linewidth}
 \begin{center}
 \epsfig{figure=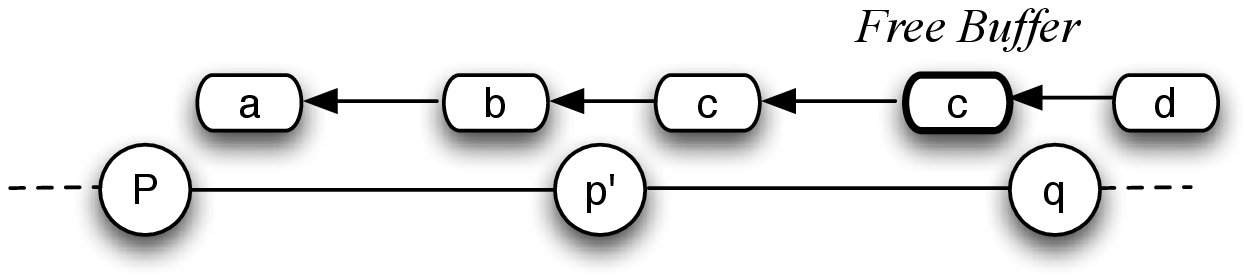,width=5.5cm}\\
  \textit{(d)} The output buffer of $q$ is free. Node $q$ can copy the message $d$.
  \end{center}
 \end{minipage}
 \caption{An example showing the free slot moving.\label{EmptySlot}}
\end{figure}

When there is a message that is in the wrong direction in the Input buffer of the processor $p_{0}$, $p_{0}$ copies
this message in its extra buffer releasing its Input buffer and it initiates a PIF wave at the same time. The aim of
the PIF waves is to escort the free slot that is in the input buffer of $p_{0}$ in order to bring it in the Output
buffer of $p_{0}$. Hence the message in the extra buffer can be copied in the output buffer to become in the right
direction. Once the PIF wave is initiated no message can be generated on this free slot, at each time the Broadcast
progresses on the chain the free slot moves as well following the PIF wave (the free slot moves by transmitting
messages on $C1$ (refer to Figure \ref{BG}). In the worst case, the free slot is the only one, hence by moving the
output buffer of the other extremity of the chain $p$ becomes free. Depending on the destination of the message that
is in the input buffer of $p$, either this message is consumed or copied in the Output buffer of $p$. In both cases
the input buffer of $p$ contains a free slot.

 In the same manner during the feedback phase, the free slot that is in the input buffer of the extremity $p$ will
progress at the same time as the feedback of the PIF wave. Note that this time the free slot moves on $C2$ (see Figure
\ref{BG}). Hence at the end of the PIF wave the output buffer that comes just after the extra buffer contains a free
slot. Thus the message that is in the extra buffer can be copied in this buffer and deleted from the extra buffer.
Note that since the aim of the PIF wave is to bring the free slot in the output buffer of $p_{0}$ then when the PIF
wave meets a processor that has a free buffer on $C2$ the PIF wave stops escorting the previous free slot and starts
the feedback phase with this second free slot (it escorts the new free slot on $C2$). Thus it is not necessary to
reach the other extremity of the chain.

Now, in the case where there is a message in the extra buffer of $p_{0}$ such as no PIF wave is executed then we are
sure that this message is an invalid message and can be deleted. In the same manner if there is a PIF wave that is
executed such that at the end of the PIF wave the Output buffer of $p_{0}$ is not free then like in the previous case
we are sure that the message that is in the extra buffer is invalid and thus can be deleted. Thus when all the buffers are full such as all the messages are different and cannot be consumed, then the extra buffer of $p_{0}$ will be released. 

Note that in the description of our algorithm, we assumed the presence of a special processor $p_{0}$. This processor has an Extra buffer used to change the direction of messages that are in the input buffer of $p_{0}$ however their destination is different from $p_{0}$. In addition it has the ability to initiate a PIF wave. Note also that the other processors of the chain do not know where this special processor is. A symmetric solution can also be used (the two processors that are at the extremity of the chain execute the same algorithm) and hence both have an extra buffer and can initiate a PIF wave. The two PIF wave initiated at each extremity of the chain use different variable and are totally independent.

\subsection{Formal description of the algorithm}

We first define in this section the different data and variables that are used in our algorithm. Next, we present the
PIF algorithm and give a formal description of the linear snap stabilizing message forwarding algorithm.

Character {\tt '?'} in the predicates and the algorithms means {\em any value}.
%%%%%%%%%%%%%%%%%%%%%%%%%%%%%Data and variables %%%%%%%%%%%%%%%%%%%%%%%%%%%%%%%%%%%%%%%

\begin{itemize}
%\paragraph{Data.}   
  \item \textbf{Data}
           \begin{itemize}
            \item $n$ is a natural integer equal to the number of processors of the chain. 
            \item $I=\{0,...,n-1\}$ is the set of processors' identities of the chain. 
            \item $N_{p}$ is the set of identities of the neighbours of the processor p.   
           \end{itemize}
  \item \textbf{Message}
            \begin{itemize}
             \item $(m,d,c)$: $m$ contains the message by itself, \ie the data carried from the sender to the recipient,
                   $ d \in I$ is the identity of the message recipient, and
                   $c$ is a color number given to the message to avoid duplicated deliveries.
             \end{itemize}

  \item \textbf{Variable}
             \begin{itemize}
              \item \textit{In the forwarding algorithm}
                        \begin{itemize}
                         \item $IN_{p}(q)$: The input buffer of $p$ associated to the link $(p,q)$.
                         \item $OUT_{p}(q)$: The output buffer of $p$ associated to the link $(p,q)$.
                         \item $EXT_{p}$: The Extra buffer of processor $p$ which is at the extremity of the chain.
                         \end{itemize}
             \item{\textit{ In the PIF algorithm}}
                          \begin{itemize}
                            \item $S_{p}= (B \vee F \vee C,q)$ refers to the state of processor $p$, $q$ is a pointer
to a neighbour of $p$.\\
                           \end{itemize}
            \end{itemize}
            
   \item \textbf{Input/Output}
           \begin{itemize}
            \item $Request_{p}$: Boolean, allows the communication with the higher layer, it is set at true by the application and false by the forwarding protocol.
             \item $PIF\mbox{-}Request_{p}$: Boolean, allows the communication between the PIF and the forwarding algorithm, it is set at true by the forwarding algorithm and false by the PIF algorithm.
             \item{The variables of the PIF algorithm are the input of the forwarding algorithm.}
            \end{itemize}

   \item \textbf{Procedure}
            \begin{itemize}
             \item $Next_{p}(d)$: refers to the neighbour of $p$ given by the routing table for the destination $d$.
             \item $Deliver_{p}(m)$: delivers the message $m$ to the higher layer of $p$. 
             \item $Choice(c)$: chooses a color for the message $m$ which is different from the color of the message 
                   that are in the buffers connected to the one that will contain $m$.
             \end{itemize}
             
    \item{\textbf{Predicate}}
             \begin{itemize}
              \item{\textit{$Consumption_{p}(q,m)$}: $IN_{p}(q)=(m,d,c)$ $\wedge$ $d=p$ $\wedge$ $OUT_{q}(p) \ne (m,d,c)$}
               \item{$leaf_{p}(q)$: $S_{q}=(B,?)$ $\wedge$ ($\forall$ $q' \in N_{p} /\{q\}$, $S_{q'}\ne(B,p)$ $\wedge$ ($consumption_{p}(q)$ $\vee$ $OUT_{p}(q')=\epsilon$ $\vee$ $OUT_{p}(q')=IN_{q'}(p)$)).}
              \item{\textit{$NO\mbox{-}PIF_{p}$}: $S_{p}=(C,NULL)$ $\wedge$ $\forall q \in N_{p}$,  $S_{q} \ne (B,?)$.}
              \item{\textit{init-PIF}: $S_{p}=(C,NULL)$ $\wedge$ ($\forall q \in N_{p}$, $S_{q}=(C,NULL)$) $\wedge$
\newline \noindent $ PIF\mbox{-}Request_{p}=true$.}
               \item{\textit{$Inter\mbox{-}trans_{p}(q)$}: $IN_{p}(q)=(m,d,c)$  $\wedge$ $d\ne p$ $\wedge$ $OUT_{q}(p) \ne IN_{p}(q)$ $\wedge$ ($\exists q' \in N_{p}/\{q\}$, $OUT_{p}(q')=\epsilon$ $\vee$ $OUT_{p}(q')=IN_{q'}(p)$).}
            %  \item{$initiator_{p}$: is the processor which is at the extremity of the chain such as $p_0$.}
              \item{$internal_{p}(q)$: $p \ne p_0$ $\wedge$ $\neg$ $leaf_{p}(q)$.} 
              \item{$Road\mbox{-}Change_{p}(m)$: $p=p_0$ $\wedge$ $IN_{p}(q)=(m,d,c)$ $\wedge$ $d \ne p$ $\wedge$ $EXT_{p}=\epsilon$ $\wedge$ $OUT_{q}(p)\ne IN_{p}(q)$.}
               \item{$\forall$ $TAction \in {C,B}$, we define $TAction\mbox{-}initiator_{p}$ the predicate: $p=p_0$ $\wedge$ (the garde of TAction in $p$ is enabled).}
               \item{$\forall$ $Tproc \in \{internal,leaf\}$ and $TAction \in \{B,F\}$, $T\mbox{-}Action\mbox{-}Tproc_{p}(q)$ is defined by the predicate: $Tproc_{p}(q)$ is true $\wedge$ TAction of $p$ is enabled.}
               %Suppose that $TProcessor$ $\in$ $\{initator_{p}, internal_{p}, leaf_{p}\}$ and $TAction$  $\in$ $\{B_{i},F_{i},C\}$ such as $S_{i}=(B,?)$. $TAction-TProcessor$ means that the action "$TAction$ of the processor of type $TProcessor$ is enabled". For instance $B_{q}-internal_{p}$ means that $\exists$ $q \in N_{p}$, $S_{q}=(B,?)$ $\wedge$ $B-action$ of the internal processor is enabled.}
               \item{$PIF\mbox{-}Synchro_{p}(q)$:  ($B_{q}\mbox{-}internal_{p}$ $\vee$ $F_{q}\mbox{-}leaf_{p}$ $\vee$ $F_{q}\mbox{-}internal_{p}$) $\wedge$ $S_{q}=(B,?)$.\\}
               \end{itemize}
               
\item{We define a fair pointer that chooses the actions that will be performed on the output buffer of a processor $p$. (Generation of a message or an internal transmission). }

\end{itemize}

%%%%%%%%%%%%%%%%%%%%%%%%% PIF Algorithm %%%%%%%%%%%%%%%%%%%%%%%%%%%%%%%%%%%%%%%%
\begin{algorithm}[htb]
\caption{PIF \label{algo:PIF}}
\begin{scriptsize}

   \begin{itemize}
       \item{\textbf{For the initiator ($p_0$)}}%%%%%%%%%%%%%%%%%%%%%%%%%%%%%%%%%%%%%%%%%%%%%%%%%
                \begin{itemize}
                \item{\textbf{B-Action::}} \textit{init-PIF }$\rightarrow$ $S_{p}:=(B,-1)$, $PIF\mbox{-}Request_{p}:=false$.
                 \item{\textbf{C-Action::}} $S_{p}=(B,-1)$ $\wedge$  $\forall q \in N_{p}$, $S_{q}=(F,?)$ $\rightarrow$ $S_{p}:=(C,NULL)$.\\
                 \end{itemize}
        \item{\textbf{For the leaf processors: $leaf_{p}(q)=true$ $\vee$ $|N_{p}|=1$}} %%%%%%%%%%%%%%%%%%%%%%%%%%%%%%%%%%%%%%%%%%%%%
                 \begin{itemize}
                  \item{\textbf{F-Action::}}  $S_{p}=(C,NULL)$  $\rightarrow$ $S_{p}:=(F,q)$.
                  \item{\textbf{C-Action::}} $S_{p}=(F,?)$ $\wedge$ $\forall q \in N_{p}$, $S_{q}=(F \vee C,?)$  $\rightarrow$ $S_{p}:=(C,NULL)$. \\
                   \end{itemize}
          \item {\textbf{For the processors}}%%%%%%%%%%%%%%%%%%%%%%%%%%%%%%%%%%%%%%%%%%
                   \begin{itemize}
                    \item{\textbf{B-Action::}}  $\exists ! q \in N_{p}$, $S_{q}=(B,?)$ $\wedge$  $ S_{p}=(C,?)$ $\wedge$  $\forall q' \in N_{p}/\{q\}$, $S_{q'} = (C,?)$ $\rightarrow$ $ S_{p}:=(B,q)$. 
                    \item{\textbf{F-Action::}} $S_{p}=(B,q)$ $\wedge$ $ S_{q}=(B,?)$ $\wedge$ $\forall q' \in N_{p}/\{q\}$, $S_{q'}=(F,?)$ $\rightarrow $ $S_{p}:=(F,q)$. 
                    \item{\textbf{C-Action::}} $S_{p}=(F,?)$ $\wedge$  $\forall q' \in N_{p}$, $S_{q'}=(F \vee C, ?)$  $\rightarrow$ $S_{p}:=(C,NULL)$. \\
                      \end{itemize}
            \item {\textbf{Correction (For any processor)}} %%%%%%%%%%%%%%%%%%%%%%%%%%%%%%%%%%%%%%%%%%
                       \begin{itemize}
                        \item {$S_{p}=(B,q)$ $\wedge$ $S_{q}=(F \vee C,?)$  $\rightarrow$ $S_{p}:=(C,NULL)$.}
                        \item{$leaf_{p}(q)$ $\wedge$ $S_{p}=(B,q)$ $\rightarrow$ $S_{p}:=(F,q)$.}
                        \end{itemize}
   \end{itemize}
   \end{scriptsize}
\end{algorithm}

%%%%%%%%%%%%%%%%%%%%%%% FORWARDING Algorithm %%%%%%%%%%%%%%%%%%%%%%%%%%%%%%%%%%
\begin{algorithm}[htb]
\caption{Message Forwarding \label{algo:MF}}
\begin{scriptsize}

   \begin{itemize}
       \item{\textbf{Message generation (For every processor) }}%%%%%%%%%%%%%%%%
       
     $\textbf{R1}$:: $Request_{p}$ $\wedge$ $ Next_{p}(d)=q$ $\wedge$ [$OUT_{p}(q)=\epsilon$ $\vee$ $OUT_{p}(q)=IN_{q}(p)$] $\wedge$ \textit{$NO\mbox{-}PIF_{p}$ }$\rightarrow$ $OUT_{p}(q):=(m,d,choice(c))$, $Request_{p}:=false$.\\

       \item{\textbf{Message consumption (For every processor) }}%%%%%%%%%%%%%%%%%
       
     $\textbf{R2}$:: \textit{$\exists q \in N_{p}$, $\exists m \in M$; $Consumption_{p}(q,m)$} $\rightarrow$ $deliver_{p}(m)$, $IN_{p}(q):=OUT_{q}(p)$.\\

        \item{\textbf{Internal transmission (For processors having 2 neighbors) }}%%%%%%%%%%%%%%%%%
 
 %        \begin{itemize}
           %  \item{
           $\textbf{R3}$:: $\exists q \in N_{p}$, $\exists m \in M$, $\exists d \in \textit{I}$; $Inter\mbox{-}trans_{p}(q,m,d)$ $\wedge$ ($NO\mbox{-}PIF_{p} \vee PIF\mbox{-}Synchro_{p}(q)$) $\rightarrow$ $OUT_{p}(q'):=(m,d,choice(c))$, $IN_{p}(q):=OUT_{q}(p)$.\\
              %\item{$\textbf{R3}$:: $\exists q \in N_{p}$, $\exists m \in M$, $\exists d \in \textit{I}$; \textit{$Inter\mbox{-}trans_{p}(q,m,d)$} $\wedge$ $PIF\mbox{-}Synchro_{p}(q)$ $\rightarrow$ $OUT_{p}(q'):=(m,d,choice(c))$, $IN_{p}(q):=OUT_{q}(p)$.}\\
          %  \end{itemize}

         \item{\textbf{Message transmission from $q$ to $p$ (For processors having 2 neighbors) }} %%%%%%%%%%
         
  %       \begin{itemize}
              %\item{
              $\textbf{R4}$:: $IN_{p}(q)=\epsilon$ $\wedge$ $OUT_{q}(p)\ne \epsilon$  $\wedge$ ($NO\mbox{-}PIF_{p} \vee PIF\mbox{-}Synchro_{p}(q)$)  $\rightarrow$ $IN_{p}(q):=OUT_{q}(p)$.\\
              %\item{$\textbf{R4}$:: $IN_{p}(q)=\epsilon$ $\wedge$ $OUT_{q}(p)\ne \epsilon$  $\wedge$ $PIF\mbox{-}Synchro_{p}(q)$  $\rightarrow$ $IN_{p}(q):=OUT_{q}(p)$.}\\
         % \end{itemize}

         \item{\textbf{Erasing a message after its transmission (For processors having 2 neighbors) }}%%%%%%%%%
         
         %\begin{itemize}
              %  \item{
              $\textbf{R5}$:: $\exists q \in N_{p}$, $OUT_{p}(q)=IN_{q}(p)$ $\wedge$ ($\forall q' \in N_{p}\setminus \{q\}$,
$IN_{p}(q')=\epsilon$) $\wedge$ ($NO\mbox{-}PIF_{p} \vee PIF\mbox{-}Synchro_{p}(q)$)   $\rightarrow$ $OUT_{p}(q):=\epsilon$, $IN_{p}(q'):=OUT_{q'}(p)$.\\
             
                %\item{$\textbf{R5}$:: $\exists q \in N_{p}$, $OUT_{p}(q)=IN_{q}(p)$ $\wedge$ $\forall q' \in N_{p}/\{q\}$, $IN_{p}(q')=\epsilon$ $\wedge$ $PIF\mbox{-}Synchro_{p}(q)$  $\rightarrow$ $OUT_{p}(q):=\epsilon$, $IN_{p}(q'):=OUT_{q'}(p)$}\\
         % \end{itemize}

         \item{\textbf{Erasing a message after its transmission (For the extremities) }}%%%%%%%%%
         
         %\begin{itemize}
              %  \item{
              $\textbf{R5'}$:: $N_{p}=\{q\}$ $\wedge$ $OUT_{p}(q)=IN_{q}(p)$ $\wedge$ $IN_{p}(q)=\epsilon$ $\wedge$
		($(p=p_0) \Rightarrow (EXT_p=\epsilon)$) 
		$\wedge$ ($NO\mbox{-}PIF_{p} \vee PIF\mbox{-}Synchro_{p}(q)$)   $\rightarrow$ $OUT_{p}(q):=\epsilon$, $IN_{p}(q):=OUT_{q}(p)$.\\
             
                %\item{$\textbf{R5}$:: $\exists q \in N_{p}$, $OUT_{p}(q)=IN_{q}(p)$ $\wedge$ $\forall q' \in N_{p}/\{q\}$, $IN_{p}(q')=\epsilon$ $\wedge$ $PIF\mbox{-}Synchro_{p}(q)$  $\rightarrow$ $OUT_{p}(q):=\epsilon$, $IN_{p}(q'):=OUT_{q'}(p)$}\\
         % \end{itemize}
        \item{\textbf{Road change (For the extremities) }}%%%%%%%%%%%%%%%
         \begin{itemize}
               \item{$\textbf{R6}$:: $Road\mbox{-}Change_{p}(m)$ $\wedge$ [$OUT_{p}(q)=\epsilon$ $\vee$ $OUT_{p}(q)=IN_{q}(p)$] $\rightarrow$ $OUT_{p}(q):=(m,d,choice(c))$, $IN_{p}(q):=OUT_{q}(p)$.}
               \item{$\textbf{R7}$:: $Road\mbox{-}Change_{p}(m)$ $\wedge$ $OUT_{p}(q) \ne \epsilon$ $\wedge$ $PIF\mbox{-}Request_{p}=false$ $\rightarrow$ $PIF\mbox{-}Request_{p}:=true$. }
               \item{$\textbf{R8}$:: $Road\mbox{-}Change_{p}(m)$ $\wedge$ $OUT_{p}(q) \ne \epsilon$ $\wedge$ $PIF\mbox{-}Request_{p}$ $\wedge$ $B\mbox{-}initiator$ $\rightarrow$ $EXT_{p}:=IN_{p}(q)$, $IN_{p}(q):=OUT_{q}(p)$.}
               \item{$\textbf{R9}$:: $p=p_0$ $\wedge$ $EXT_{p}$ $\ne$ $\epsilon$ $\wedge$ [$OUT_{p}(q)=\epsilon$ $\vee$ $OUT_{p}(q)=IN_{q}(p)$] $\wedge$ $C\mbox{-}Initiator$ $\rightarrow$  $OUT_{p}(q):=EXT_{p}$, $EXT_{p}:=\epsilon$.}
               \item{$\textbf{R10}$:: $p=p_0$ $\wedge$ $EXT_{p}$ $\ne$ $\epsilon$ $\wedge$ $OUT_{p}(q) \ne \epsilon$ $\wedge$ $OUT_{p}(q) \ne IN_{q}(p)$ $\wedge$ $C\mbox{-}Initiator$ $\rightarrow$  $EXT_{p}:=\epsilon$.}
               \item{$\textbf{R11}$:: $|N_{p}|=1$ $\wedge$ $p\ne 0$ $\wedge$ $IN_{p}(q)=(m,d,c)$ $\wedge$ $d \ne p$  $\wedge$ $OUT_{p}(q)=\epsilon$ $\wedge$ $OUT_{q}(p)\ne IN_{p}(q)$ $\rightarrow$ $OUT_{p}(q):=(m,d,choice(c))$, $IN_{p}(q):=OUT_{q}(p)$.}\\
           \end{itemize}
       \item{\textbf{Correction  (For $p_0$) }}%%%%%%%%%%%%%%%       
            \begin{itemize}
             \item{$\textbf{R12}$:: $p=p_0$ $\wedge$ $EXT_{p} \ne \epsilon$ $\wedge$ $S_{p}\ne (B,-1)$ $\rightarrow$ $EXT_{p}=\epsilon$.}
             \item{$\textbf{R13}$:: $p=p_0$ $\wedge$ $S_{p}=(B,?)$ $\wedge$ $PIF\mbox{-}Request=true$ $\rightarrow$ $PIF\mbox{-}Request=false$.}
             \item{$\textbf{R14}$:: $p=p_0$ $\wedge$ $S_{p}=(C,?)$ $\wedge$ $PIF\mbox{-}Request=true$ $\wedge$ [($IN_{p}(q)=(m,d,c)$ $\wedge$ $d=p$) $\vee$ $IN_{p}(q)=\epsilon$] $\rightarrow$ $PIF\mbox{-}Request=false$.}
            \end{itemize}       
       
  \end{itemize}
  \end{scriptsize}
  \end{algorithm}

%%%%%%%%%%%%%%%%%%%%%%%%%%%%%%%%%%%%%%%%
%%%%%%%%%%%%%%%%%%%%%%%%%%%%%%%%%%%%%%%%
\section{Proof of Correctness\label{sec:proof}}

In this section, we prove the correctness of our algorithm.
We first show that starting from an arbitrary configuration, our protocol is deadlock free.  Next, we show that no node
can be starved of generating a new message. Next, we show the snap-stabilizing property of our solution by showing that,
starting from any arbitrary configuration and even if the routing tables are not stabilized, every valid message is 
delivered to its destination once and only once in a finite time.  

Let us first state the following lemma:

\begin{lem}\label{PIF}
The PIF protocol (Algorithm 1) is snap-stabilizing.
\end{lem}  

%The proof of Lemma~\ref{PIF} is based on the fact that the PIF algorithm introduced here is similar 
%to the one proposed in \cite{Bui07}.  The only effect of the message forwarding algorithm on the PIF algorithm (w.r.t. 
%\cite{Bui07}) is that a leaf is no more only defined in terms of a topology property.  Here, a leaf is a dynamic
%property of any node. It is easy to check that this change keeps the property of snap-stabilization.

\begin{proof}
Note that the PIF algorithm introduced here is similar to the one proposed in \cite{Bui07} which is a snap stabilizing algorithm. The new thing is that we introduced the idea of dynamic leafs, processors that satisfy some properties and act like a physical leaf (they execute the F-action once they have a neighbor in a broadcast phase). Hence instead of reaching all the nodes of the chain, the PIF wave stops advancing when it meets a dynamic leaf. Note that once an internal processor $p$ executes the B-Action, it cannot execute the F-Action unless is has a neighbor $q$ such as $S_{q}=(F,p)$ (it cannot become a leaf) since to execute the F-action by any processor $p$, $S_{p}=(C,NULL)$ or for the internal processor that executes the B-Action $S_{p}=(B,q)$ ($q \in N_{p}$, $S_{q}=(B,p')$). Thus no processor becomes a dynamic leaf of the PIF wave once it executed the B-Action of the same PIF wave.
In another hand, note that the variable PIF-Request is a shared variable between the PIF algorithm and the forwarding algorithm, its role is to give the signal to the initiator to initiate the PIF wave. Hence we can deduct by analogy that the PIF algorithm proposed here is a snap stabilizing algorithm.
\end{proof}
%%%%%%%%%%%%%%%%%%%%%%%%%%%%%%%%%%%%%%%%%%%%%%%%%%%%%%%%%

We now show (Lemma~\ref{INTFree}) that the extra buffer located at $p_0$ cannot be infinitely continuously busy.  As explained in
Section~\ref{sec:algo}, this solves the problem of deadlocks. 

\begin{lem}\label{INTFree}
If the extra buffer of the processor $p_0$ ($EXT_{p_0}$) which is at the extremity of the chain contains a message then
this buffer becomes free after a finite time.
\end{lem}  

\begin{proof}
We know from Lemma~\ref{PIF} that each time $p_0$
launches a PIF wave, then this wave terminates.  When this happens, there are two cases:
\begin{enumerate}
\item{$S_{p}=(C,NULL)$}. In this case $R12$ is enabled on $p$. Since the daemon is weakly fair we are sure that $R12$ will be executed in a finite time. Thus $EXT_{p}$ will be free in a finite time too.
\item{$S_{p}=(B,?)$}. In this case, a PIF wave is executed on the chain. Note that $p$ is the initiator (it is eventually considered as the initiator). According to Lemma \ref{PIF} the PIF is a Snap stabilizing algorithm. Hence, there will be a time when $S_{q}=(F,p)$, $q \in N_{p}$. Two sub cases are possible:
\begin{itemize}
\item\label{BINT}{Either $OUT_{p}(q)=\epsilon$ or $OUT_{p}(q)=IN_{q}(p)$}. In this case $R9$ is enabled on $p$. Since the daemon is weakly fair, this rule will be executed in a finite time. Hence the message that is in the extra buffer will be copied in $OUT_{p}(q)$ and deleted from $EXT_{p}$ (see Rule $R9$). Hence $EXT_{p}=\epsilon$.
\item{$OUT_{p}(q)\ne\epsilon$ and $OUT_{p}(q)\ne IN_{q}(p)$}. Since $R10$ is enabled on $p$ and the daemon is weakly fair, $R10$ will be executed in a finite time. Thus, the message that is in the extra buffer is deleted.  
\end{itemize}
 \end{enumerate}
 
From the cases above, we deduct that in the case where the extra buffer of $p$ contains a message, then this message will be either deleted or copied in $OUT_{p}(q)$. Hence $EXT_{p}$  will be free in a finite time and the lemma holds.     
\end{proof}
%%%%%%%%%%%%%%%%%%%%%%%%%%%%%%%%%%%%%%%%%%%%%%%%%%%%%%%

 We deduce from Lemma~\ref{INTFree} that if the routing tables are not stabilized and if there is a message locking the
Input buffer of $p_0$, then this message is eventually copied in the extra buffer.  Since the latter is infinitely
often empty (Lemma~\ref{INTFree} again). 

From now on, we suppose that no generated message is deleted (we prove this property latter).

\begin{lem}\label{DeadStable}
All the messages progress in the system even if the routing tables are not stabilized.
\end{lem}

\begin{proof}
Note that if $\exists q \in N_{p}$, $IN_{p}(q)$ is free then if there is a message in $OUT_{q}(p)$, then this message is automatically copied in $IN_{p}(q)$. Hence it is sufficient to prove that the input buffer are free in a finite time. Thus Let's prove that $\forall p \in I$, when there is a message in $IN_{p}(q)$, this message is deleted from $IN_{p}(q)$ in a finite time ($q\in N_{p}$). 

%Note that once the $IN_{p}(q)$ is cleared in the processor $p$, the message in $OUT_{q}(p)$ is automatically copied in $IN_{p}(q)$ (see rules $Ri$ ( $2 \leq i \geq 9$), $R8$, $R11$). Hence it's enough to prove that $IN_{p}(q)$ is cleared in a finite time. 

Note that the input-buffers are all at an even distance from the input buffer of the processor $p_{0}$. Let define $\delta$ as the distance between the input buffer of the processor $p$ and the input buffer of processor $p_{0}$ (In the direction of the buffer graph taken in account $UT1$). The lemma is proved by induction on $\delta$. We define for this purpose the following predicate $P_{\delta}$: \\

If there is a message $mh$ in $IN_{p}(q)$ such as $IN_{p}(q)$ is at distance $\delta$ from the input buffer of $p_{0}$ then one of these two cases happens:
\begin{itemize}
\item{$m$ is consumed and hence delivered to its destination.}
\item{$m$ is deleted from the input buffer and copied either in $EXT_{p}$ or  $OUT_{p}(q')$ in a finite time.}% ($q'' \in N_{q}/\{q'\}$).}
\end{itemize}

\textbf{Initialization}. Let's prove the result for $P_{0}$. Suppose that there is a message $m$ in $IN_{p}(q)$ such as $p=p_{0}$ and $q \in N_{p}$. Two cases are possible according to the destination ($d$) of $m$:
\begin{itemize}
\item{$d=p$. In this case, since the daemon is weakly fair and since $R2$ keep being enabled on $p$ then $R2$ will be executed on $p$ in a finite time and the message $m$ in $IN_{p}(q)$ is consumed. Thus $P_{0}$ is true.}
 \item{$d\ne p$}. Since the daemon is weakly fair we are sure that $p$ will be activated. Two cases are possible:
 \begin{itemize}
\item{$OUT_{p}(q)=\epsilon$ or $OUT_{p}(q)=IN_{q}(p)$}. In this case $R6$ is enabled on $p$. Hence the message $m$ is copied in $OUT_{p}(q)$ and deleted from $IN_{p}(q)$ since a new value overwrite it (see $R6$). Thus $P_{0}$ is true.
\item{$OUT_{p}(q) \ne \epsilon$ and $OUT_{p}(q) \ne IN_{q}(p)$}. According to Lemma \ref{INTFree}, $EXT_{p}$ will be free in a finite time. In another hand since the PIF is a snap stabilizing algorithm (refer to Lemma \ref{PIF}, we are sure that the B-Action of the initiator will be enabled on $p$ in a finite time). Hence the message $m$ will be copied in this case in $EXT_{p}$ and deleted from $IN_{p}(q)$ (Note that in the case where $PIF\mbox{-}Reaquest= false$ then it is set at true (see rule $R7$)). Thus $P_{0}$ is true.
\end{itemize}
 
\end{itemize}
 
In both cases either the message is consumed or it is removed from $IN_{p}(q)$. Thus $P_{0}$ is true.\\ 

\textbf{Induction}. let $\delta$ $\geq$ $1$. We assume that $P_{2\delta}$ is true and we prove that $P_{2\delta +2}$ is true as well (Recall that the input buffers are at an even distance from the input buffer of $p_{0}$). Let $IN_{q}(p)$ be the buffer at distance $2\delta$ from the input buffer of $p_{0}$ and $IN_{p}(q')$ the one that is at distance $2\delta+2$ and contains the message $m'$. \\
In the case where the destination of $m'$ is $p$ then it will be consumed in a finite time (the daemon is weakly fair and $R2$ keep being enabled on $p$. Thus $p$ will execute $R2$ in a finite time). Hence $P_{2\delta+2}$ is true. In the other case (the destination of $m'$ is different from $p$), since $P_{2\delta}$ is true then if there is a message $m$ in $IN_{q}(p)$ then we are sure that this message will be either consumed or copied in $OUT_{q}(q'')$. Thus $IN_{q}(p)=OUT_{p}(q)$ ($OUT_{p}(q)$ is free). Two cases are possible according to the rule that is executed on $OUT_{p}(q)$ (depending on the value of the pointer on $OUT_{p}(q)$) :
\begin{enumerate}
\item\label{INTER}{$p$ executes $R3$. %or $R3$ (depending on the presence of a PIF wave or not). 
In this case the message $m'$ is copied in $OUT_{p}(q)$ and deleted from $IN_{p}(q')$ since a new value overwrite it (refer to Rule $R3$). Hence $P_{2\delta+2}$ is true.}
\item{$p$ executes $R2$ (it generates a message). Hence $OUT_{p}(q)=m''$ ($m''$ is the message generated by $p$). However, since $P_{2\delta}$ is true, then the message in $IN_{q}(p)$ will be deleted from the buffer ($q$ performs either an internal transmission or consume the message). Hence $IN_{q}(p)=OUT_{p}(q)$ in a finite time. Nevertheless, since $p$ generated a message in the previous step, we are sure that it will execute $R3$ (since the pointer on the output buffer $OUT_{p}(q)$ is fair). Thus we retrieve case \ref{INTER}.}
\end{enumerate}

From the proof above, we deduct that all the messages in the chain progress in the system and no deadlock happens even if the routing tables are corrupted.

\end{proof}

%%%%%%%%%%%%%%%%%%%%%%%%%%%%%%%%%%%%%%%%%%%%%%%%%%%%%%
Let us call a {\em valid PIF} wave every $PIF$ wave that is initiated by the processor $p_{0}$ at the same time as executing $R8$.

\begin{lem}\label{VPIF}
For every valid $PIF$ wave, when the $C\mbox{-}Action$ is executed in the initiator either $OUT_{p}(q)=IN_{q}(p)$ or $OUT_{p}(q)=\epsilon$. 
\end{lem}

\begin{proof}
The idea of the proof is as follows: 
\begin{itemize}
\item{We prove first that during the broadcast phase there is a synchrony between the PIF and the forwarding
algorithm. Note that when the message that was in the input buffer of the initiator is copied in the extra buffer, the
input buffer becomes free. The free slot in that buffer progresses in the chain at the same time as the broadcast of the PIF
wave.}
\item{Once the PIF reaches a leaf, a new buffer becomes free in $C2$ (refer to Figure \ref{BG}).
\item{As in the broadcast phase, there is a synchrony between the PIF and the forwarding algorithm
during the feedback phase. (The feedback will escort the new free slot on $C2$ to the output buffer of $p_0$.}})
\end{itemize}

Let's prove that during the broadcast phase there is a synchrony between the PIF and the forwarding algorithms.
Let's define for this purpose $\lambda$ the distance between the processor $p$ and  the processor $p_{0}$. We're going to prove the result  by induction on $\lambda$. let's define fot this purpose the predicate $P_{\lambda}$ as follow: \\

When the PIF wave is initiated and for each extra processor $p$ ($\neg leaf \wedge p \ne p_0$) that executes the B-action we have: $\exists$! $q \in N_{p}$, $S_{q}=(B,?)$, $S_{p}=(B,q)$ $\wedge$  $IN_{p}(q')=OUT_{q'}(p)$ $\wedge$ $S_{q'}=(C,NULL)$ ($q' \in N_{p}/\{q\}$).\\

\textbf{Initialization}. Let's prove that $P_{1}$ is true. Since the PIF wave is valid, when the PIF wave was
initiated, $R8$ was executed at the same time. Hence, the message that was in $IN_{p_{0}}(q)$ was copied in
$EXT_{p_{0}}$, $IN_{p_{0}}(q)=OUT_{q}(p_{0})$, $S_{p_{0}}=(B,-1)$ and $S_{q}=(C,NULL)$. Since $q$ is not a leaf only
$R3$ or $R5$ and possibly $R2$ are enabled on $q$. Note that after the execution of one of these rules
$IN_{q}(q')=OUT_{q'}(q)$ ($q' \in N_{q}/\{p_{0}\}$). In another hand $R3$ and $R5$ are not enabled only if the
$B\mbox{-}action$ of the internal processor is enabled as well. Thus when the B-Action is executed (we are sure that this
will happen since the PIF algorithm is snap stabilizing and the daemon is weakly fair) either $R3$ or $R5$ (possibly
$R2$) are executed at the same time (Recall that when two actions from the PIF and the forwarding algorithm are
enabled on the same processor at the same time they are both executed). Hence $IN_{q}(q')=OUT_{q'}(q)$,
$S_{q}=(B,p_{0})$ and $S_{q'}=(C,NULL)$. Thus $P_{1}$ is true. Note that if $R2$ is executed alone before the B-Action
then either $R3$ or $R5$ are still enabled on $q$. Hence when the B-Action is executed one of these two actions are
executed as well.

\textbf{Induction}. Let $\lambda$ $\geq$ $2$. We assume that $P_{\lambda}$ is true and we prove that $P_{\lambda +1}$ is true as well. Let $q$ and $p$ be the processors that are at distance $\lambda$ and $\lambda +1$ respectively from the processor $p_{0}$. Since $P_{\lambda}$ is true, when the B-Action of the internal processor is executed on $q$, $S_{q}=(B,q'')$, $IN_{q}(p)=OUT_{p}(q)$ and $S_{p}=(C,NULL)$. In the same manner as the proof of $P_{1}$, $R3$ or $R5$ and possibly $R2$ is enabled on $p$. Note that $R3$ or $R5$ keep being enabled unless they are executed (Note that $R1$ cannot be executed since there is a PIF wave that is executed in the neighbourhood. Thus no message is generated). In another hand $R3$ and $R5$ cannot be executed unless the B-Action is enabled as well. Hence when the B-Action is executed either $R3$ or $R5$ or $R2$ is executed at the same time. Hence $IN_{p}(q')=OUT_{q'}(p)$, $S_{p}=(B,q)$ and $S_{q'}=(C,NULL)$. Thus $P_{\lambda +1}$ is true. 

We can deduct that for the last processor $p$ that is an internal processor ($p \ne p_0$ and $\neg$ leaf), $S_{q}=(B,q'')$ and $IN_{p}(q')=OUT_{q'}(p)$, $S_{p}=(B,q)$ and  $S_{q'}=(C,NULL)$. Since $p$ is the last processor which is an internal processor then $q'$ is a leaf. Two cases are possible:
\begin{itemize}
\item{The leaf is the processor $q'$ that is at the extremity of the chain such as $q' \ne p_0$. Either $R2$ or $R11$ are enabled on $q'$. Note that the F-Action is enabled as well since $q'$ is a leaf and $S_{p}=(B,?)$. When one of these two rules is executed with the F-Action, $S_{q'}=(F,p)$ and $IN_{q'}(p)=OUT_{p}(q')$.}
\item{The leaf $q'$ is not the processor that is at the border of the chain.} In this case either $R2$ or $R3$ or $R4$ are enabled. Recall that the F-action is enabled as well. Then once the F-Action of the internal processor and one of these rules are executed, $S_{q'}=(F,p)$ and $IN_{q'}(p)=OUT_{p}(q')$.
\end{itemize}

Note that in both cases, once the leaf $q'$ executed the F-Action we have the following property: $S_{q'}=(F,p)$ and $IN_{q'}(p)=OUT_{p}(q')$. Now in the same manner that we proved the synchrony between the PIF and the forwarding algorithm during the broadcast phase. We prove the synchrony between these two algorithms during the feedback phase. The proof is by induction on $\lambda$ using the following predicate: For every internal processor $p$ that executes the F-Action, $S_{p}={F,q}$, $S_{q'}=(F,p)$, $S_{q'}=(B,?)$ and $IN_{p}(q)=OUT_{q}(p)$ ($q,q' \in N_{p}, q\ne q'$). Then when the last internal processor executes that F-action (note that the last processor is the one that is neighbour to the initiator) these properties are verified. Hence the Output buffer of the initiator is becomes free and the lemma holds.

\end{proof}

%%%%%%%%%%%%%%%%%%%%%%%%%%%%%%%%%%%%%%%%%%%%%%%%%%%%%%%
In the remainder, we say that a message is in a {\em suitable} buffer if the buffer is on the right direction to its
destination.  
A message is said to be deleted if it is removed from the system without being delivered.

Let $m$ be a message. According to Lemma \ref{DeadStable}, $m$ progresses in the system (no deadlock happens and no message stays in the same buffer indefinitely). 
So, if $m$ is in a buffer that is not suitable for it, then $m$ progresses in the system according to the buffer
graph. Thus, it eventually reaches an extremity, which changes its direction.  Now, $m$ is ensured to reach its
destination, leading to the following lemma:

\begin{lem}\label{changeR}
For every message that is not in a suitable buffer, it will undergo exactly a single route change.
\end{lem}

%%%%%%%%%%%%%%%%%%%%%%%%%%%%%%%%%%%%%%%%%%%%%%%%%%%%%%%%%%%%%%
\begin{proof}
Let $m$ be a message. 
According to Lemma \ref{DeadStable}, $m$ progresses in the system (no deadlock happens and no message stays in the same buffer indefinitely). 
So, if $m$ is not in a buffer that is not suitable for it, then $m$ progresses in the system according to the buffer
graph.
Two cases are possible:
\begin{itemize}
\item{$m$ will be in a finite time in the Input buffer of the processor $p_{0}$. Since the message is in a buffer that is not suitable for it, $p_{0}$ is not the destination of $m$. However we are sure that $EXT_{p}$ will be free in a finite time (refer to Lemma \ref{INTFree}) and that the B-Action will be enabled on $p_{0}$ in a finite time too (The PIF is a snap stabilizing algorithm). Hence the message in $IN_{p_{0}}(q)$ will be copied in $EXT_{p_{0}}$ and a PIF wave is initiated at the same time. In another hand, it has been shown in Lemma \ref{VPIF} that in the case of a valid PIF wave (Note that this is our case)  when the PIF ends ($S_{q}=(F,p_{0})$) and C-Initiator is enabled on $p_{0}$ (we are sure that this will happen since the PIF is snap stabilizing) and $OUT_{p}(q)$ becomes free. Hence $R9$ is enabled on $p_{0}$ and the message that is in the extra buffer can be put in the output buffer of $p_{0}$ and deleted from the extra buffer. Note that since the network is a chain and $p_{0}$ is at the extremity of this chain, we are sure that the message will meet its destination since it can visit all the processors. Hence no other changes route are done.}
\item{The message will reach the input buffer of the processor $p'$ that is at the other extremity of the chain ($p'\ne 0$). Since the messages progress in the system (see Lemma \ref{DeadStable}) $OUT_{p'}(q')$ will be free in a finite time. Hence when a message that is not intended to $q$ is in $IN_{q}(q')$ and since the daemon is weakly fair, we are sure that $q$ will execute $R11$ in a finite time. Thus the message will be copied in $OUT_{q}(q')$ and deleted from $IN_{q}(q')$. Now since $q$ is at the extremity of the chain, the message will meet its destination hence no other route change is performed and the lemma holds.  }
\end{itemize}
\end{proof}

 Once the routing tables are stabilized, every new message is generated in a suitable buffer.  So, it is clear from Lemma~\ref{changeR} that the number of messages that are not
in a suitable buffer strictly decreases.  The next lemma follows:

\begin{lem}\label{SWBUF}
When the routing tables are stabilized and after a finite time, all the messages are in buffers that are suitable for them.
\end{lem}  

%%%%%%%%%%%%%%%%%%%%%%%%%%%%%%%%%%%%%%%%%%%%%%%%%%%%%%%%%%%%%%
\begin{proof}
When the routing table are stabilized some of the messages still are in buffers that are not suitable for then. However, since the routing tables are stabilized, every message is generated in a suitable buffer, hence the number of 
messages that are in buffers that are not suitable for them does not increase. In another hand, According to Lemma \ref{changeR}, every message that is in the wrong direction will undergo exactly one route change and hence all the wrong messages that are in the system will be in the right direction in a finite time and the lemma holds. 

\end{proof}

%%%%%%%%%%%%%%%%%%%%%%%%%%%%%%%%%%%%%%%%%%%%%%%%%%%%%%%%%

 From there, it is important to show that any processor can generate a message in a finite time.  From
Lemma~\ref{SWBUF}, all the messages are in suitable buffers in a finite time.  Since the PIF waves are used for route
changes only, then no PIF wave will be initiated. That what we show in the two following lemmas:

\begin{lem}\label{PIFREQ}
In the case where PIF-Request=true, it will be set at false in a finite time.
\end{lem}  

\begin{proof}
Note that in the case where PIF-request is true and the B-Action of the initiator is enabled on $p_{0}$, PIF-Request will be set at false when the B-Action is executed (see B-Action of the initiator). Otherwise two cases are possible according to the state of the initiator:
\begin{itemize}
\item{$S_{p_{0}}=(B,?)$. In this case PIF-Request is set at false by the forwarding algorithm by executing $R13$ ($R13$ is enabled on $p_{0}$ and the daemon is weakly fair).
\item{$S_{p_{0}}=(C,NULL)$}. If $IN_{p_{0}}$ contains a messages and the destination of he message is not $p_{0}$ then PIF-Request will be set at false by the PIF algorithm once the PIF-wave is initiated. However in the case where the input buffer of $p_{0}$ is empty or contains a message to $p_{0}$ then $R14$ is enabled and since the daemon is weakly fair $R14$ will be executed on $p_{0}$ and hence PIF-Request is set at false.}
\end{itemize}
From the cases above we can deduct that if PIF-Request is true then it will be set at false in a finite time and the lemma holds.
\end{proof}

%%%%%%%%%%%%%%%%%%%%%%%%%%%%%%%%%%%%%%%%%%%%%%%%%%%%%%%%%
\begin{lem}\label{NOPIF}
When the routing tables are stabilized and all the messages are in suitable buffer, no PIF wave is initiated.
\end{lem}  

\begin{proof}
According to Lemma \ref{PIFREQ}. PIF-Request will be set at false in a finite time. Note that the only rule that set PIF-Request at true is $R7$. However $R7$ is never enabled since all the messages on the chain are in suitable buffer and since the routing tables are correct (all messages are generated in suitable buffer). Thus the lemma holds.
\end{proof}

%%%%%%%%%%%%%%%%%%%%%%%%%%%%%%%%%%%%%%%%%%%%%%%%%%%%%%%%%
 From this point, the fair pointer mechanism cannot be disrupted by the PIF waves anymore. So, the fairness of message
generation guarantees the following lemma:
\begin{lem}\label{GMSG}
Any message can be generated in a finite time under a weakly fair daemon.
\end{lem}

%%%%%%%%%%%%%%%%%%%%%%%%%%%%%%%%%%%%%%%%%%%%%%%%%%%%%%%%
\begin{proof}
According to Lemma \ref{NOPIF}, when the routing tables are stabilized and when all the messages are containing in
buffers that are suitable for them no PIF and no Road-change are executed. In another hand since the routing tables
are stabilized and since the buffer graph of the chain consists on two disjoint chains  (it is a DAG) then no deadlock
happens and all the messages progress in the system. Now suppose that the processor $p$ wants to generate a message. Recall that the generation of a message $m$ for the destination $d$ is always done in the output buffer of the processor $p$ connected to the link $(p,q)$ such as $Next_{p}(d)=q$. Two cases are possible
\begin{enumerate}
\item\label{free}{$OUT_{p}(q)=\epsilon$}. In this case, the processor executes either $R1$ or $R3$ in a finite time. the result of this execution depends on the value of the pointer. Two cases are possible:
\begin{itemize}
\item{the pointer refers to $R1$.} Then $p$ executes $R1$ in a finite time and we obtain the result.
\item{the pointer refers to $R3$.} Then $p$ executes $R3$ in a finite time. Hence $OUT_{p}(q)\ne\epsilon$ and we retrieve case ~\ref{Nfree}. Note that the fairness of the pointer guarantees us that this case cannot appear infinitely.
\end{itemize}
\item\label{Nfree}{$OUT_{p}(q)\ne\epsilon$.} Since all the messages move gradually in the buffer graph we are sure that $OUT_{p}(q)$ will be free in a finite time and we retrieve case  ~\ref{free}. \\
\end{enumerate} 
We can deduct that every processor can generate a message in a finite time.
\end{proof}

%%%%%%%%%%%%%%%%%%%%%%%%%%%%%%%%%%%%%%%%%%%%%%%%%%%%%%%%%%%%%%%%

Due to the color management (Function~$Choice(c)$), the next lemma follows:

\begin{lem}\label{lem:duplicate}
The forwarding protocol never duplicates a valid message even if $A$ runs simultaneously.
\end{lem}

\begin{proof}
Three cases are possible:
\begin{itemize}
\item{$m$ is in $IN_{p}(q)$. According to the rules that are enabled on $p$, three cases are possible \\ } 
\begin{itemize}
\item{the message is consumed ($R2$ is executed ) hence the message $m$ is deleted from $IN_{p}(q)$ since a new value overwrites since $IN_{p}(q)=OUT{q}(p)$ (Note that this happen only when $OUT_{q}(p)\ne IN_{p}(q))$.}
\item{$R8$ is executed . The message is copied in $EXT_{p}$ (for the processor $p_0$) and deleted from $IN_{p}(q)$ since a new value overwrites ($IN_{p}(q)=OUT{q}(p)$) in a sequential manner.}
\item{$R4$ is executed. The message is put in this case in $OUT_{p}(q')$ and deleted from $IN_{p}(q)$ in a sequential manner hence only one copy is kept ($q'$ $\in$ $N_{p}/\{q\}$). Note that these two rules are not enabled only if $OUT_{q}(p)$ does not contain the same message.}
\end{itemize}
\item{$m$ is in $OUT_{p}(q)$. In this case the message $m$ is copied in the input buffer of the processor $q$ ($IN_{q}(p)$). Hence two copies are in the system. However the message in $IN_{q}(p)$ is not consumed and not transmitted unless the copy in $OUT_{p}(q)$ is deleted (see the rules $R2$ and $R4$). } 
\item{$m$ is in $EXT_{p}$. In this case the message is either deleted or put in $OUT_{p}(q)$. Since this operation is a local operation (the copy is done between two buffer of the same processor) then the message is copied in the new buffer and deleted from the previous one in a sequential manner.}
\end{itemize}

From the cases above we can deduct that no message is duplicated in the system. Hence $m$ is delivered at most once to its destination.
\end{proof}

%%%%%%%%%%%%%%%%%%%%%%%%%%%%%%%%%%%%%%%%%%%%%%%%%%%%%%%%%%%%%%%%%%%%%%%%%%%

 From Lemma~\ref{GMSG}, any message can be generated in a finite time.  
 From the PIF mechanism and its synchronization with the forwarding protocol the only message that can be deleted is
the message that was in the extra buffer at the initial configuration.  Thus:

%%%%%%%%%%%%%%%%%%%%%%%%%%%%%%%%%%%%%%%%%%%%%%%%%%%%%%%%%%%%%%%%%%%%%%%%%

\begin{lem}\label{KeepMSG}
Every valid message (that is generated by a processor) is never deleted unless it is delivered to its destination even
if $Rtables$  runs simultaneously.
\end{lem}  

\begin{proof}
The proof is by contradiction, suppose that there is a message $m$ that is deleted without being delivered to its destination. 
%We distinguish three kinds of rules: 
%\begin{itemize}
%\item{Rules that copy the message ($R4$, $R4$). The message is only copied in the Input-Buffer of the procrssor $p$ and not deleted from the Output-buffer at the neighbor processor simultanously (the only rules that delete a message from the Output-buffer are $R5$, $R5$, $R3$ and $R3$ and the guards of these rules are not verified when $R4$, $R4$ are enabled).}.
%\item{Rules that copy and delete the message in a sequential manner}.
%\item{Rules that delete the message}.
%\end{itemize}

By construction of $R3$, this cannot be a result of an internal forwarding since the message $m$ is first of all copied in the Output-buffer $OUT_{p}(q)$ and then erased from the Input-buffer $IN_{p}(q')$ since a new value overwrites it. Note that these two rules are enabled only if $OUT_{p}(q)$=$IN_{q}(p)$ or $OUT_{p}(q)=\epsilon$. Hence when the message $m$ is copied in the $OUT_{p}(q)$ no message is deleted (one copy remains in $IN_{q}(p)$ in the case where $IN_{q}(p)=OUT_{p}(q)$). 

By the construction of Rule $R4$,  the message is only copied in the Input-Buffer and not deleted from the Output-buffer at the neighbour processor simultaneously (the only rules that delete a message from the Output-buffer are $R5$ and $R3$ and the guards of these rules are not verified when $R4$ is enabled).

If $R5$ is enabled in processes $p$, that means that $OUT_{q}(p)=IN_{p}(q)$ and $IN_{p}(q')=\epsilon$, $q'\in N_{p}/\{q\}$. When one of these two rules are enabled, $OUT_{p}(q)=\epsilon$. However according to the color management (Function Choice(c)),  we are sure that a copy of the message that was in $OUT_{p}(q)$ is in $IN_{q}(p)$. 

By the construction of the rules $R6$ and $R11$, this cannot be the result of the execution of these two rules because the message that is in $IN_{p}(q)$ such as $p_0$ and $p$ is not the destination, is copied in the Output buffer and deleted from the Input buffer sequentially and then $p$ copies the message that is in $OUT_{q}(p)$ in $IN_{p}(q)$ , so no message is deleted.

Concerning $R12$, $EXT_{p}$ such as $p_0$ contains the message $m$ and $S_{p}=(C,?)$, which means that no PIF is executed. However, for $p_0$, a message in $IN_{p}(q)$ is copied in $EXT_{p}$  (in the case where $p$ is not the destination) only if $R8$ is enabled, however, when $R8$ is enabled $B\mbox{-}intiator$ is enabled as well. Since in this case the two rules are executed at the same time, hence $S_{p}=(B,?)$. Now, for the processor $p$, Since the PIF is a valid PIF, when the $C\mbox{-}Action$ of $p$ is enabled at the same time as the rules $R10$ or $R9$. If $R10$ is executed then $EXT_{p}=\epsilon$ and $S_{p}=(C,NULL)$ (since $OUT_{p}(q)=\epsilon$ or $OUT_{p}(q)=IN_{q}(p)$), which is a contradiction, since in our case $EXT_{p}\ne \epsilon$. If $R9$ is executed, then the message in the extra buffer of $p$ ($EXT_{p}$) is copied in $OUT_{p}(q)$, $EXT_{p}$ becomes free and $S_{p}=(C,?)$, which is a contradiction with our case. Hence we are sure that the message that is in the extra buffer of $p$ is a message that was not generated by a processor. Hence when $R12$ is executed, this message is deleted (no valid message is deleted).

By the construction of the two rules $R8$ and $R9$, No valid message is deleted by the execution of the two rules, since the message is copied in the extra buffer ($R8$) or in the Output buffer ($R9$) and then it is deleted from the Input buffer ($R8$) or the extra buffer ($R9$).

Concerning the rule $R10$, according to Lemma \ref{VPIF}, when the message that is in $EXT_{p}$ is valid, when the $C\mbox{-}Action$ of the initiator is enabled either $OUT_{p}(q)=\epsilon$ or $OUT_{p}(q)=IN_{q}(p)$. However no such buffers exist. Hence the message in the extra buffer of $p$ is not a valid message (it is not generated by a processor). Hence it can be deleted.

We can deduct from all the cases above that no message that is generated by a processor is deleted, hence the lemma holds.  

\end{proof}

 %From Lemma~\ref{GMSG}, any message can be generated in a finite time.  
 %From Lemma~\ref{KeepMSG}, every valid message is never deleted unless it is delivered to its destination even if $Rtables$ runs simultaneously.
% From Lemma~\ref{changeR}, for every message that is not in a suitable buffer, it will undergo exactly a single route change if it not deleted before.
% From Lemma~\ref{lem:duplicate}, no valid message is duplicated. 
%Hence, the following theorem holds:

\begin{theorem}
The proposed algorithm (Algorithms~\ref{algo:PIF} and~\ref{algo:MF}) is a snap-stabilizing message forwarding
algorithm (satisfying SP) under a weakly fair daemon. 
\end{theorem}

\begin{proof}
From Lemma~\ref{GMSG}, any message can be generated in a finite time.  
From Lemma~\ref{KeepMSG}, every valid message is never deleted unless it is delivered to its destination even if $Rtables$ runs simultaneously.
% From Lemma~\ref{changeR}, for every message that is not in a suitable buffer, it will undergo exactly a single route change if it not deleted before.
From Lemma~\ref{lem:duplicate}, no valid message is duplicated. 
Hence, the theorem holds.
\end{proof}

%%%%%%%%%%%%%%%%%%%%%%%%%%%%%%%%%%%%%%%%%%%%%%%%%%%%%%%%%%%%%%%%%%%%%%%%%%%
%Note that for any processor $p$, the protocol delivers at most $4n-3$ invalid messages.  Indeed, in the worst
%case, initially, all the buffers are busy with different invalid messages (that were not generated).  The remark
%follows from the fact that there are only $4n-3$ buffers in the system.
%\end{lem}

%\begin{proof}
\textbf{Remarque}\\For any processor $p$, Forwarding protocol delivers at most $4n-3$ invalid messages.\\
%\end{lem}

\begin{proof}
Assume that in the initial configuration all the buffers contain a message, since these messages were not generated by the processors of the system, they are invalid messages. Suppose that the destination of the message $m$ in $IN_{p}(q)$ is the processor $p$ such as $q=0$ and $S_{q}=(B,-1)$. Suppose that the daemon activates $p$ which executes $R2$ and the F-action (it is a leaf) . Hence the message $m$ is consumed and $IN_{p}(q)=OUT_{q}(p)$. Hence $OUT_{q}(p)$ becomes free and the $C\mbox{-}action$ of the initiator is enabled, $q$ will copy then the message from $EXT_{q}$ in $OUT_{q}(p)$ and will execute the C-Action. In another hand, since there is no way to know if the messages are valid or not, they all be treated as if they are valid. Since the forwarding algorithm is snap stabilizing, all the messages that were in the buffer of the chain at the beginning are delivered. Since there is 4n-3 buffers in the system, then 4n-3 invalid messages can be delivered and the lemma holds. \\
\end{proof}

%%%%%%%%%%%%%%%%%%%%%%%%%%%%%%%%%%%%%%%%%%%%%%%%%%%%
%In order to detect the messages that have a nonexistent destination. Either the processor checks if the destination exists by searching across identities. (Recall that each processor knows all the identities of the processors that are in the chain). An other solution consists in adding a field in the message, this field will be colored by each extremity in two different colors. when a message reach an extremity with a color that is different from its color, it kows that the message was colored by the other extremity and hence explore all the chain without finding its destination. Thus it can deduct that the destination of the message does not exist . 

\section{Dynamicity\label{sec:dynamicity}}

In dynamic environments, processors may leave or join the network at any time. To keep our solution snap stabilizing we assume that there are no crashes and if a processor wants to leave the network (disconnect), it releases its buffers (it sends all the messages it has to send and wait for their reception by its neighbours) and accepts no more message before leaving.

In this discussion we assume that the rebuilt network is still a chain. It is fundamental to see that in dynamic
systems the problem of keeping messages for ghost destinations with the hope they will join the network again and the
lack of congestion are contradictory. If there is no bound on the number of leavings and joins this problem do not
admit any solution. The only way is to redefine the problem in the context of dynamicity. For example we can modify
the second point of the specification \textit{(SP)} as follows: A valid message $m$ generated by the processor $p$ to
the destination $q$ is delivered to $q$ in a finite time if $m$, $p$ and $q$ are continuously in the same connected
component during the forwarding of the message $m$. Even if that could appear very strong, this kind of hypothesis is
often implied in practice. However we can remark that this new specification is equivalent to $SP$ in static
environments. Our algorithm can easily be adapted in order to be snap stabilizing for this new specification in
dynamic chains. 

Thus we can now delete some messages as follows: we suppose that every message has an additional boolean field
initially set to false. When a message reaches an extremity which is not its destination we have two cases:
\textit{(i)} The value of the boolean is false, then the processor sets it to true and sends it in the opposite
direction. \textit{(ii)} The value of the boolean is true, then the processor deletes it (in this case, if the message
is valid, it crossed all the processors of the chain without meeting its destination).

Finally, in order to avoid starvation of some processors, the speed of joins and leavings of the processors has to be
slow enough to avoid a sequence of PIF waves that could prevent some processors to generate some messages.

\section{Conclusion\label{sec:conclu}}

In this paper, we presented the first snap-stabilizing message forwarding protocol that uses 
a number of buffers per node being independent of any global parameter.  Our protocol works on a linear chain and 
uses only $4$ buffers per link.  
It tolerates topology changes (provided that the topology remains a linear chain). 
%It is preliminary version aiming to get similar results on more general topologies (or large-scale overlays).
This is a preliminary version to get the same result on more general topologies.  
In particular, by combining a snap-stabilizing message forwarding protocol with any self-stabilizing overlay protocols 
(\eg \cite{BBKLPR10} for DHT or \cite{AspnesS2003,CDPT07,CDPT08} for \emph{tries}), we would get a solution
ensuring users to get right answers by querying the overlay architecture.

\begin{scriptsize}
\bibliographystyle{splncs}
\bibliography{PIF}
\end{scriptsize}
\end{document}